\documentclass[conference]{IEEEtran}
\usepackage{amsthm}
\usepackage{quantum}
\allowdisplaybreaks

\begin{document}
\title{Efficient achievability for quantum protocols using decoupling theorems}

\author{
    \IEEEauthorblockN{Christoph Hirche\IEEEauthorrefmark{1} and Ciara Morgan\IEEEauthorrefmark{2}}
    \IEEEauthorblockA{Institut f\"{u}r Theoretische Physik\\ Leibniz Universit\"{a}t Hannover\\ Appelstra\ss e 2, D-30167 Hannover, Germany\\
Email: \IEEEauthorrefmark{1}christoph.hirche@itp.uni-hannover.de \IEEEauthorrefmark{2}ciara.morgan@itp.uni-hannover.de}
}

\maketitle

\begin{abstract}
Proving achievability of protocols in quantum Shannon theory usually does not consider 
the efficiency at which the goal of the protocol can be achieved. Nevertheless it is known
that protocols such as coherent state merging are efficiently achievable at optimal rate. 
We aim to investigate this fact further in a general one-shot setting,  
by considering certain classes of decoupling theorems 
and give exact rates for these classes. Moreover we compare results of general 
decoupling theorems using Haar distributed unitaries with those using smaller sets of operators, 
in particular $\epsilon$-approximate 2-designs. We also observe the behavior of
our rates in special cases such as $\epsilon$ approaching zero and the asymptotic limit. 
\end{abstract}

\section{Introduction}

Quantum Shannon theory is concerned with the question of how efficiently certain resources,
such as communication and entanglement, can be used.
In this context an entire family of protocols arose \cite{WinterFam,apex},
including for example state merging and noisy teleportation.
It is necessary, for each of these protocols, to prove that a certain rate is achievable, in addition to a converse to show optimality of the rate. \\
The head of the existing family is given by the coherent state merging
protocol. In this setting Alice and Bob share a quantum state and using quantum
communication they generate maximally entangled states while Alice sends her
part of the shared state to Bob. From this all other protocols in the family can
be derived \cite{WinterFam,apex}. 
\\
In our work \cite{hircheBSC} we are especially interested in the achievability
of these protocols and how efficiently we can actually implement them.
 By efficient implementation we mean that the task can be done with maximaly polynomial use of time and space. 
Here we are looking especially at the `one-shot' setting, which means that we
apply the protocol only once.
This is more general, as we can simply take the limit of the rates to get results in 
the asymtotic case. \\ 
In general the achievability only tells us if there is a way to accomplish the
task, up to an error, at a certain rate, but is not concernd with the efficiency of doing so. For
example in the case of the coherent state merging protocol a decoupling theorem 
is used to show achievability \cite{apex,Datta}. 
Decoupling theorems tell us if it is possible to act on a state locally in such
a way that one part of the state becomes completely uncorrelated with the
rest \cite{DupuisPhD}. \\
Traditionally in the proofs of achievability the intention has been to be as general as
possible, for this reason decoupling theorems with averaging over the whole unitary group
according to the Haar measure are used. Unfortunately this requires exponential 
time. In \cite{apex} the authors remark that one could also use a
smaller set of unitaries which yield the same average as Haar distributed
unitaries such as the Clifford group of $n$ qubits, which would be achievable in
polynomial time. \\
In general we can define unitary 2-designs, that is a set of unitary matrices $U_i$
with assigned probabilities $p_i$, which yield the same average as if we had used the Haar distributed unitary group. In the same way we can define
$\epsilon$-approximate 2-designs which yield an average $\epsilon$-close to the
one of Haar distributed unitaries. Decoupling theorems using 2-designs are introduced in \cite{Oleg, OlegPaper}.
Thus we are interessted in which rates we can achieve by considering these
$\epsilon$-approximate 2-designs. \\
In section \ref{Notation} we introduce the neccessary preliminaries, including the R\'{e}nyi divergence and resulting entropic quantities. In section \ref{Decoupling} we discuss decoupling  and state a new special case for decoupling theorems using $\delta$-approximate 2-designs. In section \ref{arch} we discuss the influence of restricting ourselves to decoupling using these $\delta$-approximate 2-designs on the achievability, in particular our main result, that is efficiently achievable rates for the coherent state merging protocol. In \ref{Conclusions} we conclude. 
\section{Preliminaries}\label{Notation}
\subsection{Notation and mathematical background}
In this section we will clearify the notation used in this paper and state some
important definitons. \\
We denote a Hilbert space on a system A by $\HS_A$ and its dimension by $d_A :=
\dim(\HS_A)$. \\
In the following we will refer to set of linear operators from $\HS$ onto $\HS'$
by $\LL(\HS, \HS')$. In the case that $\HS = \HS'$ we simply write $\LL(\HS)$. 
Further $\HH(\HS)$ denotes the set of Hermitian operators and $\PP(\HS)$ are the
positive semi-definite operators. The set of quantum states will be denoted 
$S_=(\HS)$, for technical reasons we sometimes also use the set of
sub-normalized states, which we will denote by $S_\leq (\HS)$. 

For the purpose of this work it will be crucial to determine a distance between
two states. While the trace distance and the fidelity are well known, here we
will use their generalized versions, since these definitions include the case
where both states are sub-normalized. 
We define the generalized trace distance for states $\rho,\tau\in S_\leq (\HS)$ as
\begin{equation*}
\normTr{\rho -\tau} := \tr\sqrt{(\rho -\tau)^\dg (\rho -\tau)} +
\abs{\tr\rho - \tr\tau},
\end{equation*}the generalized fidelity as
\begin{equation*}
\FF(\rho,\tau) := \tr\sqrt{\left(\sqrt{\rho}\tau
\sqrt{\rho}\right)}+\sqrt{(1-\tr\rho)(1-\tr\tau)} ,
\end{equation*}
and the purified distance as
\begin{equation*}
\PP(\rho,\tau) := \sqrt{1 - \FF(\rho,\tau)^2}. 
\end{equation*}
Note that the following inequalities hold  
\begin{equation*}
\frac{1}{2}\normTr{\rho -\tau} \leq  \PP(\rho,\tau)  \leq  \sqrt{\normTr{\rho
-\tau}}.
\end{equation*}
A proof can be found in \cite{TomamichelPhD}. \\

We will also need metrics for maps in addition to states. 
For this purpose we will use the diamond norm.
\begin{defi}[Diamond norm]\label{diamond}
For $\Tau_{A\rightarrow B}$ a linear map from $\LL(\HS_A)$ to $\LL(\HS_B)$
the diamond norm is defined as
\begin{equation*}
\matnormD{\Tau_{A\rightarrow B}} := \sup_{d_R} \max_{\rho_{AR}\in\LL(\HS_{AR})}
\frac{\normTr{(\Tau_{A\rightarrow B} \otimes
\id_R)(\rho_{AR})}}{\normTr{\rho_{AR}}},
\end{equation*}
where R is a reference system.
\end{defi}

\subsection{Entropy}\label{Entropy}

Throughout this paper we will make use of many different entropic quantities,
which we will introduce in this section. 
Recently in \cite{newRenyi} and \cite{WildeWinter} a new generalization
of R\'enyi entropies was introduced, which provides a complete framework for the entropic quantities and we will
follow their definitions. 
We start by introducing the quantum R\'enyi divergence 
\begin{defi}[Quantum R\'enyi Divergence]\label{RD}
Let $\rho,\sigma\geq 0$ with $\rho\neq 0$. Then, for any
$\alpha\in(0,1)\cup(1,\infty)$, the $\alpha$-R\'enyi divergence is defined as
\begin{equation}
D_\alpha(\rho||\sigma):=\left\{\begin{array}{ll} \frac{1}{\alpha
-1}\log(\frac{1}{\tr[\rho]}&\tr[(\sigma^{\frac{1-\alpha}{2\alpha}}\rho\sigma^{\frac{1-\alpha}{2\alpha}})^\alpha]),
 \\  &\text{if}\; \rho\cancel\perp\sigma \wedge (\sigma\gg\rho\vee\alpha<1),
\\ \infty & \text{else.} \end{array}\right. 
\end{equation}
\end{defi}
From this definition we can get several quantities, such as the collision
relative entropy for $\alpha=2$
\begin{equation}
D_2(\rho||\sigma) =
\log\frac{1}{\tr[\rho]}\tr[(\sigma^{-\frac{1}{4}}\rho\sigma^{-\frac{1}{4}})^2],
\end{equation}
conditional entropies relative to a state $\sigma_B\in S_=(B)$
\begin{equation}
H_\alpha(A|B)_{\rho|\sigma} = - D_\alpha(\rho_{AB}||\id_A\otimes\sigma_B)
\end{equation}
and conditional $\alpha$-R\'enyi entropies
\begin{equation}
H_\alpha(A|B)_\rho = \sup_{\sigma_{B}\in S_= (B)} H_\alpha(A|B)_{\rho|\sigma}.
\end{equation}

These quantities have many properties, which are studied in \cite{newRenyi, WildeWinter, Beigi, Lieb}. One
we will need later is the Data-Processing Inequality (DPI) for the quantum R\'enyi
divergence, which tells us that 
\begin{equation}
D_\alpha(\rho||\sigma) \geq D_\alpha(\Epsilon(\rho)||\Epsilon(\sigma)) \label{DPI}
\end{equation}
holds for any completely positive trace preserving (CPTP) map $\Epsilon$. \\
%
%
In the following we will also use smooth min- and max-entropies, which were 
introduced by Renner in \cite{RennerPhD}. First we define the non-smooth
version.
We can retrieve these from the conditional $\alpha$-R\'enyi
entropies
\begin{eqnarray}
\Hmin(A|B)_\rho &=& \lim_{\alpha\rightarrow\infty} H_\alpha(A|B)_\rho
\nonumber\\\nonumber 
\Hmax(A|B)_\rho &=& H_\frac{1}{2}(A|B)_\rho.
\end{eqnarray} 
Based on these definitions we can now define smooth versions of the
min- and max-entropies. Therefore we define an $\epsilon$-Ball for 
$\rho\in S_\leq (\HS)$ and $0\leq\epsilon < \sqrt{\normTr{\rho}}$ as 
\begin{equation*}\label{ball}
B^\epsilon (\rho_{A}) := \{ \tau\in S_\leq (\HS): P(\tau,\rho)\leq\epsilon\} .
\end{equation*}
Note that for any $\bar\rho\in B^\epsilon (\rho_{A})$ the following holds
\begin{equation}
\normTr{\rho-\bar\rho}\leq 2 P(\rho, \bar\rho) \leq 2\epsilon.
\end{equation}
This version of the smoothed min- and max-entropies refers to \cite{TomamichelPhD}, where also many more properties of these can be found. 
\begin{defi}[Smooth min- and max-entropy]\label{smoothminmax}
For $\rho_{AB}\in S_\leq(AB)$ and $\epsilon \geq 0$, the
$\epsilon$-smooth min-entropy of A conditioned on B is defined as \\
\begin{equation*}
\Hmins(A|B)_\rho := \sup_{\rho'_{AB}\in B^\epsilon (\rho_{AB})}
\Hmin(A|B)_{\rho'}.
\end{equation*}
The $\epsilon$-smooth max-entropy of A conditioned on B is defined as \\
\begin{equation*}
\Hmaxs(A|B)_\rho := \inf_{\rho'_{AB}\in B^\epsilon (\rho_{AB})}
\Hmax(A|B)_{\rho'}.
\end{equation*}
\end{defi} 
An interessting property is given by the Asymptotic
Equipartition Property (A.E.P.) \cite{TomamichelPhD}
\begin{equation*}
\lim_{n\rightarrow\infty} \frac{1}{n} \Hmins(A|B)_{\rho^{\otimes n}} =
H(A|B)_\rho = \lim_{n\rightarrow\infty} \frac{1}{n} \Hmaxs(A|B)_{\rho^{\otimes n}}
\end{equation*}
We will also need another entropic quantity, the $H_0$-entropy, which we define as follows. 
\begin{equation*}
H_0^\epsilon(A)_\rho := \inf_{\rho'_{A}\in B^\epsilon (\rho_{A})}
\log{\tr{\Pi_{\rho_{A}}}}.
\end{equation*}
where $\Pi_{\rho_{A}}$ is a projector onto the support on $\rho_A$.\\

\section{Decoupling theorems}\label{Decoupling}

Decoupling theorems are an important set of tools for proving rates,
especially for the achievability. For example in \cite{Datta} a 
decoupling theorem was used to prove the achievability of their rates for the coherent state
merging. 
The usual scenario is as follows. We want a system A to be decoupled, up to some
error, from the environment E, meaning that the joint state $\rho_{AE}$, after sending Alice
part of the state trough a channel, can be written in product form
$\tau_A\otimes\rho_E$, where $\tau$ is the maximally mixed state.
Usually decoupling is achieved by first applying local operations, mostly represented by unitaries 
distributed according to the Haar-measure on the A part
of a state $\rho_{AR}$, followed by sending this part of the state through a channel
to another system B. The goal is to leave the system R unchanged in the
process and the introduced error as small as possible. \\
There are various versions of decoupling theorems, in particular
using a partial trace or a projective measurement as special cases of the
channel between the two systems A and B. Later a theorem for decoupling with
arbitrary maps was introduced by Dupuis \cite{DupuisPhD} and then stated again in terms of smooth min-entropies in \cite{OSDec}, where also a converse for the general case was
given. \\
Recently the decoupling theorem has been further generalized by Szehr in
\cite{Oleg, OlegPaper} to the case where an approximate unitary 2-design is used instead of Haar-distributed
unitaries. \\
Here we will focus on decoupling theorems using $\delta$-approximate unitary 2-designs.
The motivation is that while we need exponential time to achieve
Haar-distributed unitaries, it is shown in \cite{2dc} that approximate unitary 2-designs can be 
implemented with random circuits of length $O(n(n + \log \frac{1}{\delta}))$.
Recently it was even shown that these can be implemented with a circuit of
polynomial length just using diagonal circuits, which theoretically allows us to
apply all gates at the same time \cite{yoshi1, yoshi2}. \\
We start by introducing the neccessary tools. 
First we define unitary k-designs. 
Therefore we use the abbreviations
\begin{eqnarray}
\EE_U[\rho] := \sum_i p_i\, U_i^{\otimes k}\, \rho\, (U_i^\dg)^{\otimes k} \\
\EE_H[\rho] := \int_U U^{\otimes k}\, \rho\, (U^\dg)^{\otimes k}\, dU.
\end{eqnarray}
\begin{defi}[$\delta$-approximate unitary k-design]\label{kdes}
Let $\DD$ be a set of unitary matrices $U_i$ on $\HS$ which have a probability
$p_i$ for each $U_i$, with the constraint $\sum_i p_i = 1$. $\DD$ is called a
unitary k-design if and only if 
\begin{equation*}
\EE_U[\rho]  =  \EE_H[\rho]
\end{equation*}
and a $\delta$-approximate unitary k-design if
\begin{equation*}
\MatnormD{\EE_U[\rho]  - \EE_H[\rho] } \leq \delta
\end{equation*}
holds for $\rho\in\LL(\HS^{\otimes k}).$
\end{defi}
Since coherent state merging is the most general quantum protocol, we
would like to apply the decoupling theorem in \cite{Oleg} to its achievability. 
For the rates in \cite{Datta} a version of the decoupling theorem using
Haar distributed unitaries with the partial trace instead of the general
CP-map is used. 
The direct approach would therefore be to simply plug in the trace for the
CP-map. Unforunately this way leads us to rates which allow a slightly 
bigger error in the process of decoupling than the version used in \cite{Datta}
and therefore we can not recover the rates when we look at the case
when $\delta$ approaches zero. 
For this reason we state a decoupling theorem here which is more specialized to
the partial trace map. 
\begin{thm}[Smoothed decoupling for $\delta$-approximate 2-designs]\label{Dec4} 
For a small enough $\epsilon >0$ and $\rho_{AR}\in
S_{\leq}(\HS_{AR})$, while $\log(d_{A_1}) \leq \frac{1}{2}
[\Hmins(A|R)_\rho +\log(d_A)] + \log\left(\epsilon
\frac{1}{\sqrt{1+3\frac{|A|^2}{|A_1|}\delta}}\right)$ the following
holds\\
\begin{equation}
\sum_{(p_i, U_i)\in\DD} p_i
\normTr{\tr_{A_2}(U_i\rho_{AR}U_i^\dg)-\tau^{A_1}\otimes\rho^R} \leq 5\epsilon \label{left}
\end{equation}
where the pairs $(p_i, U_i)$ are such that $\DD$ constitutes an
$\delta$-approximate 2-design. 
\end{thm}
The proof can be found in subsection \ref{proof}. \\ 
Note, that if we choose $\delta$ to be zero, we recover the exact rates of the
decoupling theorem in \cite{Datta}. 
This is a reasonable expectation due to the definition of unitary 2-designs.

\subsection{Lemmata}
Two Lemmata for use in the following proof. 
\begin{lem}\label{c1}
Let $A=A_{1}A_{2}$. Then
\begin{eqnarray*}
\sum_{(p_i, U_i)\in\DD} p_i (U\otimes U)^\dagger
(\id_{A_2A_{2'}}\otimes F_{A_1A_{1'}}) (U\otimes U)\\ 
 \leq
\frac{1}{|A_1|}\id_{AA'} + \frac{1}{|A_2|}F_{AA'} + |A| \delta \id_{AA'}
\end{eqnarray*}
where the pairs $(p_i, U_i)$ are such that $\DD$ constitutes an
$\delta$-approximate 2-design.  
\end{lem}
\begin{proof}    
To prove the lemma above, we make direct use of the original Lemma C.1 in
\cite{BertaReverse} and the properties of $\delta$-approximate 2-designs. \\
Let $\DD'$ be an exact 2-design, by the definition of 2-designs and Lemma C.1 in
\cite{BertaReverse} we know 
\begin{eqnarray*}
\sum_{(p_i, U_i)\in\DD'} p_i (U\otimes U)^\dagger
(\id_{A_2A_{2'}}\otimes F_{A_1A_{1'}}) (U\otimes U) \\
 \leq
\frac{1}{|A_1|}\id_{AA'} + \frac{1}{|A_2|}F_{AA'}.
\end{eqnarray*}
Using the infinity-norm we can now find an upper bound on the maximal eigenvalue
difference when using an $\delta$-approximate 2-design \cite{Low}.
\begin{equation}
\matnormI{\EE_U[\rho] - \EE_H[\rho]} \leq |A| \matnormD{\EE_U[\rho] - \EE_H[\rho]} \leq |A| \delta .
\end{equation}
Combining both steps Lemma \ref{c1} follows.
\end{proof} 

\begin{lem}\label{c2}
Let $\rho_{AR}\in\PP(\HS_{AR}), A=A_{1}A_{2}$ and
$\sigma_{A_1R}(U)=\tr_{A_2} [
(U\otimes\id_R)\rho_{AR}(U\otimes\id_R)^\dagger ]$. Then
\begin{eqnarray*}
\sum_{(p_i, U_i)\in\DD} p_i \tr\sigma_{A_1R}(U_i) \\
\leq
\frac{1}{|A_1|}\tr[\rho^2_R] + \frac{1}{|A_2|}\tr[\rho^2_{AR}] + |A| \delta
\tr[\rho^2_R]
\end{eqnarray*}
\end{lem}
\begin{proof}
The proof follows that of Lemma C.2 in
\cite{BertaReverse}, by replacing Lemma C.1 from
\cite{BertaReverse} with Lemma \ref{c1}. 
\end{proof}
 
\subsection{Proof of decoupling theorem (Theorem \ref{Dec4})} \label{proof}

We prove our decoupling Theorem (\ref{Dec4}) using $\delta$-approximate 2-designs. 
Some steps are analogous to the proof of the decoupling theorem in \cite{Datta}, 
but for completeness we will state all steps here. \\
For fixed $0 < \epsilon\leq 1, \bar{\rho}^{AR}\in B^\epsilon(\rho^{AR})$ and
$\bar\sigma_{A_1R}(U)=\tr_{A_2} [
(U\otimes\id_R)\bar\rho_{AR}(U\otimes\id_R)^\dagger ]$, 
we aim for an upper bound on the expectation value on the left hand side of
Equation (\ref{left}). First we use the triangle inequalitiy
\begin{eqnarray}
&& \normTr{\sigma_{A_1R}(U_i) - \tau^{A_1}\otimes\rho^R} \nonumber\\\nonumber
& \leq &\normTr{\sigma_{A_1R}(U_i) - \bar\sigma_{A_1R}(U_i)} +
\normTr{\bar\sigma_{A_1R}(U_i) - \tau^{A_1}\otimes\bar\rho^R}  \\\nonumber
&& + \normTr{\tau^{A_1}\otimes\bar\rho^R - \tau^{A_1}\otimes\rho^R} \\\nonumber
& \leq &
2\normTr{\sigma_{A_1R}(U_i) - \bar\sigma_{A_1R}(U_i)} +
\normTr{\bar\sigma_{A_1R}(U_i) - \tau^{A_1}\otimes\bar\rho^R} ,
\end{eqnarray}
where we also use that $\bar\sigma^R(U_i) = \bar\rho^R$ and monotonicity of the
trace distance under partial trace. \\
We continue by upper-bounding the first term. 
\begin{eqnarray*}
&&\sum_{(p_i, U_i)\in\DD} p_i \normTr{\sigma_{A_1R}(U_i) -
\bar\sigma_{A_1R}(U_i)} \\
&=& \sum_{(p_i, U_i)\in\DD} p_i \normTr{\tr_{A_2} [
(U\otimes\id_R)(\rho_{AR}-\bar\rho_{AR})(U\otimes\id_R)^\dagger ]} \\
&\leq & \sum_{(p_i, U_i)\in\DD} p_i \normTr{(U\otimes\id_R)(\rho_{AR}-\bar\rho_{AR})(U\otimes\id_R)^\dagger} \\
&=& \sum_{(p_i, U_i)\in\DD} p_i \normTr{\rho_{AR}-\bar\rho_{AR}} = \normTr{\rho_{AR}-\bar\rho_{AR}} \leq  2\epsilon,
\end{eqnarray*}
where the first inequality again follows from the monotonicity of the
trace distance under partial trace. The second equality follows from
the invariance of the trace distance under applying unitary operators. 
Now we make use of the inequality (Lemma 5.1.3 of \cite{RennerPhD})
\begin{equation}
\normTr{H} \leq \sqrt{\tr{\Omega}} \normTwo{\Omega^{-\frac{1}{4}} H
\Omega^{-\frac{1}{4}}},
\end{equation}
where H is an operator in $\HH(\HS)$ and $\Omega\in S_=(\HS)$, to upper bound the second term. Then
\begin{eqnarray}
&& \sum_{(p_i, U_i)\in\DD} p_i
\normTr{\bar\sigma_{A_1R}(U_i) - \tau^{A_1}\otimes\bar\rho^R}  \nonumber\\
&\leq& \sum_{(p_i, U_i)\in\DD} p_i \sqrt{\tr{\Omega}}
\sqrt{\normTwo{\Omega^{-\frac{1}{4}} ( \bar\sigma_{A_1R}(U_i) -
\tau^{A_1}\otimes\bar\rho^R) \Omega^{-\frac{1}{4}}}^2} \nonumber\\
&\leq& \sqrt{\tr{\Omega}}
\sqrt{\sum_{(p_i, U_i)\in\DD} p_i \normTwo{\Omega^{-\frac{1}{4}} (
\bar\sigma_{A_1R}(U_i) - \tau^{A_1}\otimes\bar\rho^R) \Omega^{-\frac{1}{4}}}^2}. \nonumber\\
\label{root}
\end{eqnarray}
The second inequalitiy follows from concavity of the squareroot function. 
We choose $\Omega = \id_{A_1}\otimes\omega_R,\, \omega_R\in S_=(\HS_R)$ and define
\begin{eqnarray*}
\tilde\sigma_{A_1R}(U_i) &=& \Omega^{-\frac{1}{4}} \bar\sigma_{A_1R}(U_i)
\Omega^{-\frac{1}{4}} \\
&=& \tr_{A_2} [
(U\otimes\id_R)\tilde\rho_{AR}(U\otimes\id_R)^\dagger ],
\end{eqnarray*} 
where $\tilde\rho_{AR} =
(\id_{A}\otimes\omega_R^{-\frac{1}{4}})\bar\rho_{AR}(\id_{A}\otimes\omega_R^{-\frac{1}{4}})$.
Note that $\tilde\rho_{R} = \tr_A \tilde\rho_{AR} = \omega_R^{-\frac{1}{4}}
\bar\rho_{R}\omega_R^{-\frac{1}{4}}$. \\
We concentrate on the term in the second squareroot in (\ref{root}). 
\begin{eqnarray}
&&\sum_{(p_i, U_i)\in\DD} p_i \normTwo{\Omega^{-\frac{1}{4}} (
\bar\sigma_{A_1R}(U_i) - \tau^{A_1}\otimes\bar\rho^R) \Omega^{-\frac{1}{4}}}^2
\nonumber\\
&=& \sum_{(p_i, U_i)\in\DD} p_i  \normTwo{\tilde\sigma_{A_1R}(U_i) -
\tau^{A_1}\otimes\tilde\rho^R}^2 \nonumber\\
&=& \sum_{(p_i, U_i)\in\DD} p_i  \tr[(\tilde\sigma_{A_1R}(U_i) -
\tau^{A_1}\otimes\tilde\rho^R)^2] \nonumber\\
&=& \sum_{(p_i, U_i)\in\DD} p_i  (\tr[\tilde\sigma_{A_1R}(U_i)^2] \label{bino}\\\nonumber
&& -2\tr[\tilde\sigma_{A_1R}(U_i)(\tau^{A_1}\otimes\tilde\rho^R)]) +
\tr[(\tau^{A_1}\otimes\tilde\rho^R)^2]) 
\end{eqnarray}
For now, we focus on the second term in (\ref{bino}). 
\begin{eqnarray}
&& -\sum_{(p_i, U_i)\in\DD}
p_i 2\tr[\tilde\sigma_{A_1R}(U_i)(\tau^{A_1}\otimes\tilde\rho^R))]
\nonumber\\
&=& - 2\tr[\sum_{(p_i, U_i)\in\DD}
p_i \tilde\sigma_{A_1R}(U_i)(\tau^{A_1}\otimes\tilde\rho^R))] \nonumber\\
&\leq& -
2\tr[(((\frac{1}{|A_1|}-|A|\delta)\id_{A_1})\otimes\tilde\rho^R)(\tau^{A_1}\otimes\tilde\rho^R))]  \nonumber\\
&=& -2\tr[(\tau^{A_1}\otimes\tilde\rho^R)^2] + 2|A||A_1|\delta\tr[(\tau^{A_1}\otimes\tilde\rho^R)^2],\nonumber\\
 \label{sec}
\end{eqnarray}
where the inequality follows from the properties of $\delta$-approximate unitary
2-designs. 
Plugging (\ref{sec}) into (\ref{bino}) we get
\begin{eqnarray}
&& \sum_{(p_i, U_i)\in\DD} p_i (\tr[\tilde\sigma_{A_1R}(U_i)^2]  \nonumber\\
&&-2\tr[\tilde\sigma_{A_1R}(U_i)(\tau^{A_1}\otimes\tilde\rho^R)] +
\tr[(\tau^{A_1}\otimes\tilde\rho^R)^2]) \nonumber\\
&\leq& \sum_{(p_i, U_i)\in\DD} p_i (\tr[\tilde\sigma_{A_1R}(U_i)^2]  \nonumber\\
&& -(1-2|A_1||A|\delta)\tr[(\tau^{A_1}\otimes\tilde\rho^R)^2])
\nonumber\\
&=& \sum_{(p_i, U_i)\in\DD} p_i (\tr[\tilde\sigma_{A_1R}(U_i)^2]  \nonumber\\
&& -(\frac{1}{|A_1|}-2|A|\delta)\tr[(\tilde\rho^R)^2]) \nonumber\\
&\leq& \frac{1}{|A_2|}\tr[(\tilde\rho^{AR})^2] + 
3|A|\delta\tr[(\tilde\rho^R)^2]
\nonumber\\
&=&\frac{1}{|A_2|} 2^{D_2(\bar\rho^{AR}||\id_A\otimes\omega_R)}
+ 3|A|\delta2^{D_2(\bar\rho^{R}||\omega_R)}
\nonumber\\
&\leq& (\frac{1}{|A_2|} + 3|A|\delta)
2^{D_2(\bar\rho^{AR}||\id_A\otimes\omega_R)} 
\nonumber\\
&=& (\frac{1}{|A_2|} + 3|A|\delta)
2^{-H_2(A|R)_{\bar\rho|\omega}} 
\nonumber\\
&\leq& (\frac{1}{|A_2|} + 3|A|\delta)
2^{-\Hmin(A|R)_{\bar\rho|\omega}} \label{Hmin}
\end{eqnarray}
where the second inequality follows from Lemma (\ref{c2}), the second equality
from the definition of the collision-entropy, the third inequality from the DPI in (\ref{DPI}), the fourth inequality from Lemma B.6 in
\cite{BertaReverse}.
\\
Now we choose $\bar\rho^{AR}\in S_\leq(\HS_{AR})$ and $\omega_R\in S_=(\HS_R)$
such that $\Hmin(A|R)_{\bar\rho|\omega} = \Hmins(A|R)_{\rho}$. \\
Using this and (\ref{Hmin}), we get
\begin{eqnarray}
&&\sum_{(p_i, U_i)\in\DD} p_i \normTr{\sigma_{A_1R}(U_i) -
\tau^{A_1}\otimes\rho^R}  \nonumber\\   \label{clo}
&\leq& 4\epsilon + \sqrt{|A_1|}\sqrt{(\frac{1}{|A_2|} + 3|A|\delta)
2^{-\Hmins(A|R)}} .
\end{eqnarray}
To finish the proof we need to show when the right hand side of (\ref{clo}) can
be upper bounded by $\epsilon$. This holds when
\begin{eqnarray*}
\sqrt{|A_1|}\sqrt{(\frac{1}{|A_2|} + 3|A|\delta) 2^{-\Hmins(A|R)}} \label{close} &\leq \epsilon.
\end{eqnarray*}
By rearanging we get
\begin{eqnarray*}
\log|A_1| - \frac{1}{2}(\Hmins(A|R)+\log|A|) +& \\
\log(1 + 3\frac{|A|^2}{|A_1|}\delta)) &\leq \log\epsilon
\end{eqnarray*}
Thus we can conclude that 
\begin{equation}
\sum_{(p_i, U_i)\in\DD} p_i \normTr{\sigma_{A_1R}(U_i) -
\tau^{A_1}\otimes\rho^R} \leq 5\epsilon 
\end{equation}
holds, if 
\begin{equation}
\log|A_1| - \frac{1}{2}(\Hmins(A|R)+\log|A|) \leq \log\frac{\epsilon}{\sqrt{1 +
3\frac{|A|^2}{|A_1|}\delta}}
\end{equation}
and therefore Theorem (\ref{Dec4}) is proven. 

\section{Application to information theoretic protocols}\label{arch}

We now want to apply our results to the achievability of information theoretic
protocols. Therefore we take a closer look at the
one-shot coherent state merging protocol. We start with a single copy of a tripartite pure state
$\ket{ABR}{\Psi}$.
The goal is to transfer the A part of the state to the party which also holds
the B part and at the same time generate entaglement between these parties,
while leaving the reference system unchanged. We call this an $\epsilon$-error
protocol if this goal is achieved with an error not bigger then $\epsilon$. \\
Using our results from the last section, we can now follow the proof of
 the achievability theorem for the coherent state merging protocol
in \cite{Datta}(Theorem 8) now using the decoupling theorem given by Theorem
\ref{Dec4}. We get the following achievability theorem. 
\begin{thm}[Achievability of coherent state merging using
$\delta$-approximate 2-designs]\label{achiev2} 
For any tripartite pure state
$\ket{ABR}{\Psi}$ there exists an $\epsilon$-error one-shot coherent state merging protocol with an entanglement gain
$e^{(1)}_\epsilon$ and a quantum communication cost $q^{(1)}_\epsilon$ bounded
by
\begin{equation*}
e_\epsilon^{(1)} \geq \frac{1}{2}[\Hmins[\lambda](A|R)_\Psi +H^\lambda_0(A)_\Psi]
+ \log(\lambda' \frac{1}{\sqrt{1+3\frac{|A|^2}{|A_1|}\delta}})
\end{equation*}
\begin{equation*}
q_\epsilon^{(1)} \leq \frac{1}{2}[-\Hmins[\lambda](A|R)_\Psi +H^\lambda_0(A)_\Psi]
- \log(\lambda' \frac{1}{\sqrt{1+3\frac{|A|^2}{|A_1|}\delta}})
\end{equation*}
for $0 < \epsilon\leq 1$ and $\lambda >0$ such that $\epsilon = 2\sqrt{5\lambda'} + 2\sqrt{\lambda}$
and $\lambda' = \lambda + \sqrt{4\sqrt{\lambda}-4\lambda}$.
\end{thm}
Note, that  in the case of $\delta$ going to zero, we recover the theorem in \cite{Datta}. 
Also we can reproduce the asymtotic IID rates
\begin{eqnarray}
q^\infty =& \lim_{\epsilon\rightarrow 0} \lim_{n\rightarrow \infty}
q^{(1)}_{\epsilon,n} =& \frac{1}{2} I(A:R)_\Psi, \\
e^\infty =& \lim_{\epsilon\rightarrow 0} \lim_{n\rightarrow \infty}
e^{(1)}_{\epsilon,n} =& \frac{1}{2} I(A:B)_\Psi,
\end{eqnarray}
for any $\delta$.
 
\section{Conclusion}\label{Conclusions}
In this work we explored the effect of limiting ourselves to \emph{efficiently
implementable} tools for one-shot quantum Shannon theory on the rates for
information theoretic protocols. \\
We showed achievable rates and that they are consistent with the previously
known results. 
To get this consistency we derived a new special case of the decoupling theorem.

\section*{Acknowledgment}
The authors would like to thank R. F. Werner and Y. Nakata for fruitful discussions. \\
This work was supported by the EU grants SIQS and QFTCMPS and by the cluster of excellence EXC 201 Quantum Engineering and Space-Time Research.

\bibliographystyle{IEEEtran}
\bibliography{Thesis}

\end{document}